\documentclass[12pt,a4paper]{amsart}
\usepackage[headings]{fullpage}
\usepackage{amsmath,amssymb}
\usepackage[hidelinks]{hyperref}
\numberwithin{equation}{section}

\theoremstyle{plain}
\newtheorem{theorem}{Theorem}[section]
\newtheorem{thm}{Theorem}

\newtheorem{lemma}[theorem]{Lemma}
\newtheorem{proposition}[theorem]{Proposition}

\theoremstyle{definition}
\newtheorem{remark}[theorem]{Remark}

\newcommand{\BC}{\mathbb C}
\newcommand{\CP}{\mathcal P}
\newcommand{\CS}{\mathcal S}
\newcommand{\CO}{\mathcal O}
\DeclareMathOperator{\Tr}{Tr}
\DeclareMathOperator{\BetaFn}{B}

\title[Purity, entropy and spectral-volume decay for APPT states]{Spectral
Optimization for Absolutely PPT States: Purity, Entropy, and Volume Decay}

\author{Anh T. Tran}
\address{Department of Mathematical Sciences, The University of Texas at Dallas, Richardson, TX 75080, USA}
\email{att140830@utdallas.edu}


\hypersetup{
  pdftitle={Spectral Optimization for Absolutely PPT States: Purity, Entropy, and Volume Decay},
  pdfauthor={Anh T. Tran},
  pdfsubject={Spectral geometry of absolutely PPT states},
  pdfkeywords={absolutely separable states, absolutely PPT states, purity, von Neumann entropy, relative spectral volume}
}

\begin{document}
\allowdisplaybreaks
\raggedbottom

\begin{abstract}
We study maximum purity and minimum von Neumann entropy of absolutely positive
partial transpose (APPT) states, together with the relative volume of their
spectral sets. A consequence of Hildebrand's
criterion gives an explicit outer spectral polytope, whose vertices we
classify. Optimizing purity over this polytope, together with the Song--Chen
result for the $2\otimes3$ system, yields, for every $mn\ge6$, an explicit
upper bound on the purity of APPT states that is asymptotic to $4/(3mn)$.
This improves the previous $2/(mn)$ upper bound for absolutely separable
states.  The new bound applies to every APPT state, a class containing all
absolutely separable states, and is sharp for every $2\otimes n$ system with
$n\ge3$. For every $3\otimes n$ system with $n\ge3$, however, the unique
outer-polytope maximizer is not APPT, so the bound is strict and disproves the
D\~ung--Kh\^oi qutrit--qudit conjecture. The polytope also gives an explicit
entropy lower bound in arbitrary bipartite dimensions and, together with the
Song--Chen extreme-point classification, the exact minimum entropy for every
$2\otimes n$ system.
Finally, exact formulas for the relative volumes of an inner polytope and the
outer spectral polytope give explicit two-sided bounds on the qubit--qudit
relative spectral volume $a_n$ whose ratio is less than $4$ and tends to $3$.
Consequently,
$a_n=\Theta\!\left(\sqrt n(4/27)^n\right)$, and the relative volume of the
qubit--qudit APPT spectral set (equivalently, the absolutely separable spectral
set) has exact exponential decay rate $\ln(27/4)$.
\end{abstract}
\maketitle

\section{Introduction and main results}

Let $2\le m\le n$.  A bipartite quantum state $\rho$ on
$\BC^m\otimes\BC^n$ is called separable if it can be written as a
convex combination of product states, namely,
$\rho =\sum_i \omega_i \, \rho_i^A \otimes \rho_i^B$, where
$\omega_i \ge 0$, $\sum_i \omega_i=1$, and $\rho_i^A$ and $\rho_i^B$
are density operators on
subsystems $A$ and $B$, respectively. Peres \cite{Pe} showed that any separable
state $\rho$ has a positive partial transpose (PPT), meaning that the
density matrix of $\rho$ remains positive semidefinite after taking the
partial transpose with respect to one of its subsystems.

Following \cite{KZ} and \cite{Hi}, we say that a state $\rho$ on
$\BC^m\otimes\BC^n$ is absolutely separable (resp. APPT) if it remains
separable (resp. PPT) under every global unitary transformation; that
is, $W\rho W^\dagger$ is separable (resp. PPT) for every unitary
operator
$$
W:\BC^m\otimes\BC^n\longrightarrow\BC^m\otimes\BC^n.
$$
It follows from
\cite{Pe} that the set of all absolutely separable states is a subset
of the set of all APPT states.  In the two-qubit case, the corresponding
spectral characterization was obtained by Verstraete, Audenaert and De Moor
\cite{VAD}.  Johnston \cite{Jo} completed the proof that the two sets are
identical for all qubit--qudit systems, namely, when $m=2$.  In higher
bipartite dimensions, whether the two sets coincide remains open \cite{AJR}.

In spectral terms, absolute separability asks whether changing only the global
eigenbasis of a state can generate entanglement from a fixed spectrum.  APPT is
a tractable relaxation that rules out the generation of a state whose partial
transpose is not positive semidefinite under any such change of eigenbasis.
Purity and entropy locate the least-mixed robust spectra, while relative
spectral volume measures how common these spectra are in the probability
simplex.

Hildebrand \cite{Hi} gave a necessary and sufficient spectral criterion for
the APPT property.  It consists of linear matrix inequalities in the
eigenvalues of $\rho$.  One consequence is the following inequality for every
APPT state $\rho$ on $\BC^m\otimes\BC^n$:
\begin{equation} \label{ineq}
\lambda_1 \le \lambda_{mn-1} + 2\sqrt{\lambda_{mn-2} \lambda_{mn}}.
\end{equation}
Throughout, $\lambda_1 \ge \lambda_2 \ge \cdots \ge \lambda_{mn}$ denote the
eigenvalues of the density matrix of $\rho$ in decreasing order.  In the
displayed spectra below, $\underbrace{z,\ldots,z}_{k}$ denotes $k$ repeated
entries equal to $z$; if $k=0$, the block is omitted.  Inequality \eqref{ineq}
and $2\sqrt{ab}\le a+b$ imply
$$
\lambda_1\le \lambda_{mn-2}+\lambda_{mn-1}+\lambda_{mn}.
$$
Consequently, every ordered APPT spectrum belongs to the polytope defined by
this last inequality together with nonnegativity, decreasing order, and
normalization.  We call it the outer spectral polytope.  The full criterion
in \cite{Hi} is more complicated in general.  For qubit--qudit systems
($m=2$), however, the single inequality \eqref{ineq} is necessary and
sufficient for the APPT property.

It is important to distinguish an upper bound that holds for every absolutely
separable state from a sufficient condition that guarantees absolute
separability.  Gurvits and Barnum \cite{GB} proved that a state $\rho$ on
$\BC^m\otimes\BC^n$ is separable if its purity satisfies
\begin{equation} \label{ineq-purity}
\Tr(\rho^2) \le \frac{1}{mn-1}.
\end{equation}
Since unitary transformations preserve purity, every state satisfying
\eqref{ineq-purity} is absolutely separable.  Thus \eqref{ineq-purity}
describes a known subset of absolutely separable states; it is not an upper
bound on the purity of every absolutely separable state.  Finding a necessary
and sufficient condition for absolute separability in general remains a
challenging problem.  More recent sufficient criteria in arbitrary dimensions
were obtained by Abellanet-Vidal et al.~\cite{AV+} using inverses of linear
maps and convex optimization.

Xiong and Sze \cite[Proposition 4]{XS} obtained the APPT purity bound
$\Tr(\rho^2)\le10/(mn+9)$.  An earlier result of Jivulescu et
al.~\cite[Proposition 8.2]{J+} states that every APPT state satisfies
$\lambda_1\le3/(mn+2)$.  Since $\Tr(\rho^2)\le\lambda_1$, it also gives
the APPT purity bound $3/(mn+2)$.  For absolutely separable states, Filippov et
al.~\cite[Proposition 1]{FMJ} obtained
$\Tr(\rho^2)\le9/(mn+8)$, and Kondra et al.~\cite[Lemma 5]{K+}
subsequently proved $\Tr(\rho^2)\le2/(mn)$ when
$mn\ge6$. Because every absolutely separable state is APPT, both APPT bounds
also apply to absolutely separable states. For comparison as bounds on
absolutely separable states, when $mn\ge6$ the four estimates are ordered
as
$$
\frac{2}{mn}<\frac{3}{mn+2}<\frac{9}{mn+8}<\frac{10}{mn+9}.
$$
Thus $2/(mn)$ is the strongest among these four dimension-independent
estimates.

For qubit--qudit systems ($m=2$), Song and Chen characterized the
extreme points of the APPT state set, which in these dimensions coincides
with the absolutely separable state set: the two-qubit case is treated in
their Theorem~9 and the case $n\ge3$ in their Theorem~10 \cite{SC}.
Comparing these extreme points, they obtained the maximum purity: it is
$3/8$ for $n=2$
\cite[Corollary~11(ii)]{SC}, while for $n\ge3$ it is $2/(3n)$ when
$n$ is even and $(6n+4)/(3n+1)^2$ when $n$ is odd
\cite[Corollary~11(iv)]{SC}.  Theorem~\ref{main-thm} independently
recovers these formulas for $n\ge4$ through the present optimization
argument; the case $n=3$ follows from Song and Chen.  They also found
the two-qubit minimum von Neumann entropy \cite[Corollary~11(ii)]{SC} and reported
numerical candidates in dimensions $2\otimes3$ and $2\otimes4$
\cite[Remark following Corollary~11]{SC}.  We determine the minimum
von Neumann entropy rigorously in every qubit--qudit dimension.  The $2\otimes3$
candidate is confirmed; for $2\otimes4$, the minimizing spectrum is
instead
$$
\frac{1}{14}
(\underbrace{3,\ldots,3}_{3},\underbrace{1,\ldots,1}_{5}).
$$

For the uniform measure on the spectral simplex used below, Slater
\cite[Eqs.~(27) and (30)]{Sl} evaluated the two-qubit absolutely separable
volume exactly.  Writing $a_n$ for the relative spectral volume of the
qubit--qudit APPT set, we derive explicit lower and upper bounds valid for
every $n\ge2$ and determine its exact exponential decay rate.

Ahiable, Kothakonda and Winter
\cite[Theorem 6.10 and Proposition 6.14]{AKW} computed the minimum von Neumann entropy
and the exact relative volume of an inner polytope, denoted $\CP_{m,n}$ in their
work.  Their
comparison with the full APPT volume uses Monte Carlo estimates and
log-linear regression \cite[Figures~10--11 and Table~1]{AKW}; for $m=2$,
they report the fitted base-$10$ slope $-0.7625$ over the tested range.  We
determine the qubit--qudit entropy minimum without the inner polytope.  For
volume, the exact inner- and outer-polytope formulas give, respectively,
rigorous lower and upper bounds for $a_n$.  The ratio of the upper bound to
the lower bound is less than $4$ and tends to $3$.  Thus these bounds
determine the exact exponential decay rate $\ln(27/4)$.

In addition to the three main theorems below, Proposition
\ref{leading-sums-prop} gives sharp leading-partial-sum bounds on the outer
polytope for $1\le k\le mn-3$; these bounds are attained by qubit--qudit APPT
spectra.  These partial sums are the Ky Fan eigenvalue sums and quantify how
strongly a spectrum is concentrated in its largest entries.

We first maximize purity over the outer spectral polytope.  For $mn\ge8$,
the following result is the sharpest purity bound that follows from
$\lambda_1\le\lambda_{mn-2}+\lambda_{mn-1}+\lambda_{mn}$ together with
ordering and normalization.  The remaining total dimension $mn=6$ in the
theorem is supplied by Song--Chen, as explained in the proof.

\begin{thm}\label{main-thm}
Let $2\le m\le n$ and $mn\ge6$. The purity of any absolutely PPT
(APPT) state $\rho$ on $\BC^m\otimes\BC^n$ satisfies
$$
\Tr(\rho^2)\le\begin{cases}
\displaystyle\frac{4}{3mn} & \text{if $mn \equiv 0 \pmod{4}$}, \medskip \\

	                 \displaystyle\frac{4(3mn-2)}{(3mn-1)^2} & \text{if $mn \equiv 1 \pmod{4}$}, \medskip \\
	                 
	                 \displaystyle\frac{4(3mn + 4)}{(3mn+2)^2}  & \text{if $mn \equiv 2 \pmod{4}$}, \medskip\\
	                 
	               \displaystyle\frac{4(3mn + 2)}{(3mn+1)^2} & \text{if $mn \equiv 3 \pmod{4}$}. 
\end{cases}
$$
Moreover, equality is attained when $m=2$ and $n\ge3$, whereas the inequality
is strict when $m=3$ and $n\ge3$.
In particular, the displayed upper bound is strictly smaller than $2/(mn)$
for every $mn\ge6$ and is asymptotic to $4/(3mn)$ as
$mn\longrightarrow\infty$.
\end{thm}

\begin{remark}
For qutrit--qudit systems ($m=3$), D\~ung and Kh\^oi conjectured that the
expression appearing in Theorem~\ref{main-thm} is the exact maximum APPT
purity and specified the same spectra as the outer-polytope maximizers
described in Theorem~\ref{linear}
\cite[Conjecture~3.2]{DK}.  In each such dimension, the conjectured spectrum
is the unique purity maximizer over the outer polytope, but it is not APPT.
Hence the inequality in
Theorem~\ref{main-thm} is strict for every $3\otimes n$ system with $n\ge3$,
and the D\~ung--Kh\^oi conjecture is false in those dimensions.  Recent work of
Wang, Chen and Song \cite{WCS} gives a detailed classification of full-rank
$3\otimes3$ APPT extreme points having exactly three distinct eigenvalues,
but does not settle the global optimization problem.  Determining the exact
maximum APPT purity remains open.
\end{remark}

We next state the entropy and volume consequences for qubit--qudit
systems.  In this paper, entropy is measured in bits.  For a state $\rho$ and a
probability vector $\pmb\lambda$, set
$$
S(\rho)=-\Tr(\rho\log_2\rho),
\qquad
H(\pmb\lambda)=-\sum_j\lambda_j\log_2\lambda_j,
$$
with the convention $0\log_2 0=0$.  If $\pmb\lambda(\rho)$ denotes the
eigenvalue vector of $\rho$, then
$S(\rho)=H(\pmb\lambda(\rho))$.  Thus $S$ denotes entropy on states, whereas
$H$ denotes the corresponding spectral entropy.  Natural logarithms, denoted
by $\ln$, are used in intermediate estimates and for exponential decay rates.
For $n\ge4$, let $r_n$ be the unique integer in
$$
\left\{\left\lfloor\frac n2(3\ln3-2)\right\rfloor,
\left\lceil\frac n2(3\ln3-2)\right\rceil\right\}
$$
that minimizes
\[
h_n(r)=\log_2(2n+2r)-\frac{3r}{2n+2r}\log_2 3.
\]

Let $\Lambda^{\mathrm{APPT}}_{m,n}$ and
$\Lambda^{\mathrm{ASEP}}_{m,n}$ denote the permutation-invariant sets
of spectra of APPT and absolutely separable states, respectively; a
superscript $\downarrow$ denotes the corresponding subset of decreasingly
ordered spectra.

\begin{thm}\label{entropy-main-thm}
For every $n\ge2$, the minimum von Neumann entropy of APPT states on
$\BC^2\otimes\BC^n$ is
$$
\min_{\substack{\rho\text{ APPT on}\\
\BC^2\otimes\BC^n}}S(\rho)
=
\begin{cases}
\log_2 3,&n=2,\\[1mm]
\displaystyle
\log_2(14+8\sqrt2)
-\frac{4(3+2\sqrt2)}{14+8\sqrt2}\log_2(3+2\sqrt2),&n=3,\\[2mm]
h_n(r_n),&n\ge4.
\end{cases}
$$
The unique decreasingly ordered minimizing eigenvalue vectors in the three
cases are, respectively,
$$
\frac13\bigl(\underbrace{1,\ldots,1}_{3},0\bigr),\qquad
\frac{1}{14+8\sqrt2}
\bigl(\underbrace{3+2\sqrt2,\ldots,3+2\sqrt2}_{4},
\underbrace{1,\ldots,1}_{2}\bigr),
$$
and, for $n\ge4$,
$$
\frac{1}{2n+2r_n}
\Big(\underbrace{3,\ldots,3}_{r_n},
\underbrace{1,\ldots,1}_{2n-r_n}\Big).
$$
Equivalently, in each case the minimizing eigenvalue vector is unique up to
permutation.
The same statement holds for absolutely separable states.
\end{thm}

\begin{thm}[Exponential decay]
\label{volume-main-thm}
For every $n\ge2$, let $a_n$ denote the relative Euclidean volume of
$\Lambda^{\mathrm{APPT}}_{2,n}$ in the probability simplex
$\Delta_{2n-1}$. Set
$$
p_n=\frac{n!(2n-2)!}{(3n-2)!},\qquad
q_n=\frac{n(2n-1)}{n-1}
\bigl[\BetaFn(n,2n-2)-\BetaFn(2n-1,2n-2)\bigr],
$$
where $\BetaFn$ is the beta function, and let $c_n=q_n/p_n$. Then
$$
p_n\le a_n\le q_n=c_np_n<4p_n\quad(n\ge2),
\qquad \lim_{n\to\infty}c_n=3.
$$
Consequently,
$a_n=\Theta\!\left(\sqrt n(4/27)^n\right)$, and its exact exponential decay
rate is
$$
\lim_{n\to\infty}-\frac1n\ln a_n=\ln\frac{27}{4}.
$$
Since the APPT and absolutely separable spectral sets coincide for
$2\otimes n$, the same conclusions hold for absolutely separable spectra.
\end{thm}

\begin{remark}[Scope of the entropy and volume bounds]
For $m\ge3$, the outer-polytope results give an entropy lower bound and a
relative spectral-volume upper bound, but not the exact APPT quantities.
\end{remark}

\section{The outer spectral polytope: purity and leading partial sums}

For any integer $d \ge 4$, let $\CS_d$ denote the outer spectral
polytope consisting of real sequences
$\pmb{\lambda}=(\lambda_1,\lambda_2,\ldots,\lambda_d)$ such that
$\lambda_1\ge\lambda_2\ge\cdots\ge\lambda_d\ge0$,
$\sum_{i=1}^d\lambda_i=1$, and
$\lambda_1\le\lambda_{d-2}+\lambda_{d-1}+\lambda_d$.

Define $F$ on $\CS_d$ by
$F(\pmb{\lambda})=\sum_{i=1}^{d}\lambda_i^2$.  We maximize
$F(\pmb{\lambda})$ over $\CS_d$.  Since $F$ is strictly convex and
$\CS_d$ is a compact convex polytope, the maximum exists.  Moreover,
every maximizer is a vertex: if a maximizer were a nontrivial convex
combination of two distinct points of $\CS_d$, strict convexity would
force one of those points to have a larger value of $F$.

\begin{lemma}[Vertices of the outer polytope]
\label{outer-vertices-lemma}
Let $d \ge 4$. The vertices of $\CS_{d}$ are given as follows.
\begin{itemize}
\item Type $1$ vertices:
$\frac1i\bigl(\underbrace{1,\ldots,1}_{i},
\underbrace{0,\ldots,0}_{d-i}\bigr)$, where
$d-2 \le i \le d$. \smallskip

\item Type $2$ vertices:
$\frac1{d-1+i}\bigl(\underbrace{2,\ldots,2}_{i},
\underbrace{1,\ldots,1}_{d-1-i},0\bigr)$, where
$1 \le i \le d-3$. \smallskip

\item Type $3$ vertices:
$\frac1{d+2i}\bigl(\underbrace{3,\ldots,3}_{i},
\underbrace{1,\ldots,1}_{d-i}\bigr)$, where
$1 \le i \le d-3$.
\end{itemize} 
\end{lemma}

\begin{proof}
Let $\CO_d$ be the set of decreasing probability vectors in
$\mathbb R^d$.  To see its simplex structure explicitly, set
$\lambda_{d+1}=0$ and
$$
x_i=i(\lambda_i-\lambda_{i+1}),\qquad 1\le i\le d.
$$
Then $x_i\ge0$, $\sum_i x_i=1$, and
$$
\pmb\lambda=\sum_{i=1}^d x_i\pmb\nu_i,
\qquad
\pmb{\nu}_i = \frac1i\bigl(\underbrace{1,\ldots,1}_{i},
\underbrace{0,\ldots,0}_{d-i}\bigr),\qquad 1\le i\le d.
$$
Thus $\CO_d=\operatorname{conv}\{\pmb\nu_1,\ldots,\pmb\nu_d\}$ is a
$(d-1)$-simplex.  In particular, each segment
$\overline{\pmb\nu_i\pmb\nu_j}$ with $i\ne j$ is an edge.

Let $G(\pmb{\lambda}) = \lambda_{d-2} + \lambda_{d-1} +\lambda_{d} - \lambda_1$. Then $\CS_{d}=\{ \pmb{\lambda} \in \CO_{d} \mid G(\pmb{\lambda}) \ge 0\}$. 
Write $g_i=G(\pmb\nu_i)$.  Then
$$
g_i = \begin{cases}
	  - \displaystyle\frac{1}{i} & \text{$1 \le i \le d-3$}, \smallskip \\
	  
            0 & \text{if $i=d-2$},  \smallskip \\
            
            \displaystyle\frac{1}{d-1}  & \text{if $i = d-1$}, \smallskip \\
            
            \displaystyle\frac{2}{d}  & \text{if $i = d$}.
\end{cases}
$$
Hence $\pmb{\nu}_i \in \CS_{d}$ if and only if $d-2 \le i \le d$.

We now determine the vertices of $\CS_d$ that are not vertices of
$\CO_d$.  In the coordinates $x_i$ above, the additional condition is
$\sum_i g_i x_i\ge0$.  A new vertex must satisfy
$\sum_i g_i x_i=0$.  If three or more of its coordinates $x_i$ were positive,
there would be a nonzero vector $\pmb h$ supported on those coordinates such that
$$
\sum_i h_i=0,\qquad \sum_i g_i h_i=0.
$$
For sufficiently small $\varepsilon>0$, both $\pmb x+\varepsilon\pmb h$ and
$\pmb x-\varepsilon\pmb h$ would remain nonnegative and satisfy the two
equalities.  The corresponding point of $\CS_d$ would then be the midpoint of
two distinct points of $\CS_d$, contrary to its being a vertex.  Thus a new
vertex cannot have three or more positive coordinates.  A point with only one
positive coordinate is an original simplex vertex, so a genuinely new vertex
has exactly two positive coordinates.  Their $g_i$-values must have
opposite signs, so every new vertex has the form
$$
\pmb{\lambda}=t\pmb{\nu}_i+(1-t)\pmb{\nu}_j,
\qquad 0<t<1,
$$
where $1\le i\le d-3$ and $d-1\le j\le d$. By linearity of $G$,
$$
G(\pmb{\lambda}) = \begin{cases}
            \displaystyle\frac{1-t}{d-1} - \frac{t}{i}
            & \text{if $j = d-1$}, \smallskip \\
            
            \displaystyle\frac{2(1-t)}{d} - \frac{t}{i}
            & \text{if $j = d$}.
\end{cases}
$$
Solving $G(\pmb\lambda)=0$ gives
$$
t = \begin{cases}
            \displaystyle\frac{i}{d-1+i}  & \text{if $j = d-1$}, \smallskip \\
            
            \displaystyle\frac{2i}{d+2i}    & \text{if $j = d$}.
\end{cases}
$$
Substitution into $\pmb\lambda=t\pmb\nu_i+(1-t)\pmb\nu_j$
gives
$$
\pmb{\lambda}  = \begin{cases}
             \displaystyle\frac1{d-1+i}\bigl(\underbrace{2,\ldots,2}_{i},
             \underbrace{1,\ldots,1}_{d-1-i},0\bigr)
             & \text{if $j = d-1$}, \smallskip \\
             \displaystyle\frac1{d+2i}\bigl(\underbrace{3,\ldots,3}_{i},
             \underbrace{1,\ldots,1}_{d-i}\bigr)
             & \text{if $j = d$}.
\end{cases}
$$ 
Each displayed point is indeed a vertex: in the $x_i$-coordinates, it is the
unique point determined by $x_r=0$ for $r\notin\{i,j\}$ together with
$\sum_r x_r=1$ and $\sum_r g_r x_r=0$.
Together with the surviving vertices $\pmb\nu_{d-2}$,
$\pmb\nu_{d-1}$, and $\pmb\nu_d$, these are all the vertices of $\CS_d$.
\end{proof}

\begin{theorem}[Maximum purity on the outer polytope] \label{linear}
If $d \ge 8$, then the maximum of  $F(\pmb{\lambda}) $ on $\CS_{d}$ is given by
$$
\begin{cases}
\displaystyle\frac{4}{3d} & \text{if $d \equiv 0 \pmod{4}$}, \medskip \\

	                 \displaystyle\frac{4(3d-2)}{(3d-1)^2} & \text{if $d \equiv 1 \pmod{4}$}, \medskip \\
	                 
	                 \displaystyle\frac{4(3d + 4)}{(3d+2)^2}  & \text{if $d \equiv 2 \pmod{4}$}, \medskip\\	              
	                 
	               \displaystyle\frac{4(3d + 2)}{(3d+1)^2} & \text{if $d \equiv 3 \pmod{4}$}.
\end{cases}
$$
Moreover, this maximum  is attained only at
$$
\begin{cases}
		   \displaystyle\pmb{\lambda} = \frac2{3d}
		   \bigl(\underbrace{3,\ldots,3}_{d/4},
		   \underbrace{1,\ldots,1}_{3d/4}\bigr)
       & \text{if $d \equiv 0 \pmod{4}$ and $d>8$},  \smallskip \\
	   
	   \displaystyle\pmb{\lambda} \in
	       \bigl\{\frac16(\underbrace{1,\ldots,1}_{6},
	       \underbrace{0,\ldots,0}_{2}),
	       \frac1{12}(\underbrace{3,\ldots,3}_{2},
	       \underbrace{1,\ldots,1}_{6})\bigr\}
       & \text{if $d=8$}, \smallskip \\
	   
             \displaystyle\pmb{\lambda} =
             \frac2{3d-1}
             \bigl(\underbrace{3,\ldots,3}_{(d-1)/4},
             \underbrace{1,\ldots,1}_{(3d+1)/4}\bigr)
             & \text{if $d \equiv 1 \pmod{4}$}, \smallskip \\
             
             \displaystyle\pmb{\lambda} =
             \frac2{3d+2}
             \bigl(\underbrace{3,\ldots,3}_{(d+2)/4},
             \underbrace{1,\ldots,1}_{(3d-2)/4}\bigr)
             & \text{if $d \equiv 2\pmod{4}$}, \smallskip \\
                      
             \displaystyle\pmb{\lambda} =
             \frac2{3d+1}
             \bigl(\underbrace{3,\ldots,3}_{(d+1)/4},
             \underbrace{1,\ldots,1}_{(3d-1)/4}\bigr)
             & \text{if $d \equiv 3\pmod{4}$}.         
\end{cases}
$$
\end{theorem}
\begin{proof}
Recall that every maximizer of $F$ on $\CS_d$ is a vertex of $\CS_d$.

\smallskip
\noindent\emph{Type $1$.}
At type $1$ vertices
$\pmb{\lambda}=\frac1i\bigl(\underbrace{1,\ldots,1}_{i},
\underbrace{0,\ldots,0}_{d-i}\bigr)$,
where $d-2 \le i \le d$, we have $F(\pmb{\lambda}) = \frac{1}{i}$.
So the maximum at type $1$ vertices is $M_1 = \frac{1}{d-2}$. 

\smallskip
\noindent\emph{Type $2$.}
At type $2$ vertices
$\pmb{\lambda}=\frac1{d-1+i}
\bigl(\underbrace{2,\ldots,2}_{i},
\underbrace{1,\ldots,1}_{d-1-i},0\bigr)$, where $1 \le i \le d-3$,
we have $F(\pmb{\lambda})=f_2(i)$, where
$$
f_2(x)=\frac{d-1+3x}{(d-1+x)^2},
\qquad
f_2'(x)=\frac{d-1-3x}{(d-1+x)^3}.
$$
Thus the continuous maximum occurs at $x=(d-1)/3$, and the discrete
type $2$ maximum satisfies
\begin{equation}\label{type2-continuous-bound}
M_2:=\max_{1\le i\le d-3}f_2(i)
\le \frac{9}{8(d-1)}.
\end{equation}

\smallskip
\noindent\emph{Type $3$.}
At type $3$ vertices
$\pmb{\lambda}=\frac1{d+2i}
\bigl(\underbrace{3,\ldots,3}_{i},
\underbrace{1,\ldots,1}_{d-i}\bigr)$,
where $1\le i\le d-3$, we have $F(\pmb\lambda)=f_3(i)$, where
$$
f_3(x)=\frac{d+8x}{(d+2x)^2},
\qquad
f_3'(x)=\frac{4(d-4x)}{(d+2x)^3}.
$$
Hence $f_3$ increases up to $d/4$ and decreases thereafter.  If
$d=4q+r$ with $r\in\{1,2,3\}$, the only possible maximizing integers are
$q$ and $q+1$, and direct subtraction gives
$$
\operatorname{sgn}\bigl(f_3(q+1)-f_3(q)\bigr)
=\operatorname{sgn}\bigl((6r-12)q+r(r-1)\bigr).
$$
This sign is negative for $r=1$ and positive for $r=2,3$.  Together with
the case $r=0$, where $d/4=q$, this shows that the unique maximizing index
is
$$
i_\star=\begin{cases}
d/4,&d\equiv0\pmod4,\\
(d-1)/4,&d\equiv1\pmod4,\\
(d+2)/4,&d\equiv2\pmod4,\\
(d+1)/4,&d\equiv3\pmod4.
\end{cases}
$$
For $d\ge8$, all these indices lie in $\{1,\ldots,d-3\}$, and
substitution gives
$$
M_3=f_3(i_\star)=\begin{cases}
\displaystyle\frac{4}{3d},&d\equiv0\pmod4,\smallskip\\
\displaystyle\frac{4(3d-2)}{(3d-1)^2},&d\equiv1\pmod4,\smallskip\\
\displaystyle\frac{4(3d+4)}{(3d+2)^2},&d\equiv2\pmod4,\smallskip\\
\displaystyle\frac{4(3d+2)}{(3d+1)^2},&d\equiv3\pmod4.
\end{cases}
$$
Substituting $i_\star$ into the type $3$ spectrum also gives the spectra
listed in the statement of the theorem.

We now compare the three vertex types without treating the four cases modulo
$4$ separately.  Since $|d-4i_\star|\le2$ and $d+2i_\star>d$, we have
$$
M_3=\frac{d+8i_\star}{(d+2i_\star)^2}
=\frac{4}{3d}
-\frac{(d-4i_\star)^2}{3d(d+2i_\star)^2}
>\frac{4(d^2-1)}{3d^3}.
$$
This lower bound is already sufficient. Indeed,
$$
\frac{4(d^2-1)}{3d^3}-\frac{9}{8(d-1)}
=\frac{5d^3-32d^2-32d+32}{24d^3(d-1)}>0
\qquad(d\ge8),
$$
whereas
$$
\frac{4(d^2-1)}{3d^3}-\frac{1}{d-2}
=\frac{d^3-8d^2-4d+8}{3d^3(d-2)}>0
\qquad(d\ge9).
$$
For the first numerator, set $d=8+s_1$; for the second, set
$d=9+s_2$.  Here $s_1,s_2\ge0$, and the resulting polynomials are,
respectively,
$$
5s_1^3+88s_1^2+416s_1+288,
\qquad
s_2^3+19s_2^2+95s_2+53.
$$
All their coefficients are positive, so both numerators are positive. By
\eqref{type2-continuous-bound} and $M_1=1/(d-2)$, it follows that
$M_3>M_2$ for every $d\ge8$ and $M_3>M_1$ for every $d\ge9$.
For $d=8$, direct substitution gives $M_3=M_1=1/6$, attained only at
the type $3$ vertex with $i=2$ and the type $1$ vertex with $i=6$,
respectively.  This proves all the assertions for $d\ge8$.
\end{proof}

The same defining inequality also gives sharp bounds for sums of the largest
$k$ entries of a spectrum.

\begin{proposition}[Leading partial sums]\label{leading-sums-prop}
Let $d\ge4$ and $1\le k\le d-3$ be integers.  Every
$\pmb\lambda\in\CS_d$ satisfies
\begin{equation}\label{leading-sum-bound}
\sum_{j=1}^k\lambda_j\le\frac{3k}{d+2k},
\qquad
\lambda_k\le\frac3{d+2k}.
\end{equation}
Both bounds are sharp on $\CS_d$, with equality at the type~$3$
vertex
$$
\frac1{d+2k}
\bigl(\underbrace{3,\ldots,3}_{k},
\underbrace{1,\ldots,1}_{d-k}\bigr).
$$
Consequently, if $d=mn$ with $2\le m\le n$, then
\eqref{leading-sum-bound} holds for every
$\pmb\lambda\in(\Lambda^{\mathrm{APPT}}_{m,n})^\downarrow$.
When $m=2$, both bounds are sharp for APPT spectra and, equivalently,
for absolutely separable spectra.
\end{proposition}

\begin{proof}
Set $s=\sum_{j=1}^k\lambda_j$.  Since the entries are decreasing,
$\lambda_1\ge s/k$.  Moreover, the defining inequality of $\CS_d$
gives
$$
\lambda_1\le\lambda_{d-2}+\lambda_{d-1}+\lambda_d
\le3\lambda_{d-2},
$$
and hence $\lambda_{d-2}\ge\lambda_1/3$.  Every term in the middle sum
below is therefore at least $\lambda_1/3$, while the sum of the last three
terms is at least $\lambda_1$.  It follows that
\begin{align*}
1
&=s+\sum_{j=k+1}^{d-3}\lambda_j
  +\lambda_{d-2}+\lambda_{d-1}+\lambda_d\\
&\ge s+\frac{d-k-3}{3}\lambda_1+\lambda_1\\
&=s+\frac{d-k}{3}\lambda_1\\
&\ge s+\frac{d-k}{3k}s
=\frac{d+2k}{3k}s.
\end{align*}
This proves the first inequality in \eqref{leading-sum-bound}; the second
follows from $\lambda_k\le s/k$.  The displayed type~$3$ vertex attains
equality in both inequalities.

For every factorization $d=mn$, the inclusion
$(\Lambda^{\mathrm{APPT}}_{m,n})^\downarrow\subseteq\CS_d$ gives the
stated APPT consequence.  Suppose now that $m=2$, so $d=2n$.  Since
$k\le d-3$, the last three entries of
the displayed spectrum are all $1/(d+2k)$, and therefore
$$
\lambda_{d-1}+2\sqrt{\lambda_{d-2}\lambda_d}
=\frac3{d+2k}=\lambda_1.
$$
The exact qubit--qudit criterion \eqref{ineq} shows that this spectrum
is APPT, and Johnston's equivalence \cite{Jo} shows that it is
absolutely separable.
\end{proof}

For $k=1$, Jivulescu et al.~\cite[Proposition~8.2]{J+} proved that the
value $3/(d+2)$ is attained by APPT and absolutely separable spectra for
every bipartition $d=mn$.  Song and Chen
\cite[Corollary~17]{SC} determined the exact corresponding leading-sum
maxima in the $3\otimes3$ APPT state set.  Proposition
\ref{leading-sums-prop} instead gives a uniform outer bound in every
bipartite dimension and is sharp for every qubit--qudit system.  The
restriction $k\le d-3$ is necessary: for $k=d-2$, the type~$1$ vertex
$\pmb\nu_{d-2}\in\CS_d$ has leading sum $1$, whereas
$3(d-2)/(3d-4)<1$.

\begin{proof}[Proof of Theorem \ref{main-thm}]
The case $mn=6$, necessarily $(m,n)=(2,3)$, follows from the maximum
$11/50$ proved by Song and Chen \cite[Corollary~11(iv)]{SC}; Johnston's
qubit--qudit equivalence \cite{Jo} additionally gives the same maximum for
absolutely separable states.  We may therefore assume
$mn\ge8$. Since
$\lambda_1 \le \lambda_{mn-2}+\lambda_{mn-1}+\lambda_{mn}$ for any
APPT state $\rho$ on $\BC^m\otimes\BC^n$, the purity
$\Tr(\rho^2)$ is bounded above by the maximum of $F(\pmb{\lambda})$
on $\CS_{mn}$.  The upper bound in Theorem \ref{main-thm} now follows
from Theorem \ref{linear} by taking $d=mn$.

We first prove strictness in qutrit--qudit systems.  Let $m=3$, $n\ge3$,
and $d=3n$.  Since $d\ge9$, Theorem~\ref{linear} shows that the outer
maximum is attained at a unique type~$3$ spectrum
$$
\pmb{\lambda}^{*}=\frac1{d+2i}
\bigl(\underbrace{3,\ldots,3}_{i},
\underbrace{1,\ldots,1}_{d-i}\bigr),
$$
where $i=i_\star$; the displayed formulas in that theorem give $i\ge2$ and
$d-i\ge6$.
Hildebrand's qutrit criterion \cite[Corollary~V.3]{Hi} requires, in
particular,
$$
K(\pmb{\lambda})=
\begin{pmatrix}
2\lambda_d & \lambda_{d-1}-\lambda_1 & \lambda_{d-3}-\lambda_2\\
\lambda_{d-1}-\lambda_1 & 2\lambda_{d-2} & \lambda_{d-4}-\lambda_3\\
\lambda_{d-3}-\lambda_2 & \lambda_{d-4}-\lambda_3 & 2\lambda_{d-5}
\end{pmatrix}
\succeq0.
$$
Here $K(\pmb\lambda)\succeq0$ means that $K(\pmb\lambda)$ is positive
semidefinite.
We have $\lambda_1^*=\lambda_2^*=3/(d+2i)$ and
$\lambda_{d-5}^*=\cdots=\lambda_d^*=1/(d+2i)$, while
$\lambda_3^*=1/(d+2i)$ exactly when $i=2$.
At $\pmb{\lambda}^{*}$, the scaled matrix
$(d+2i)K(\pmb{\lambda}^{*})$ has diagonal entries $2$.  Its $(1,2)$-
and $(1,3)$-entries are $-2$, while its $(2,3)$-entry is $0$ if $i=2$
and $-2$ if $i\ge3$.  Thus
$$
\det\bigl((d+2i)K(\pmb{\lambda}^{*})\bigr)=
\begin{cases}
-8,&i=2,\\
-32,&i\ge3
\end{cases}.
$$
In both cases the determinant is negative, so $K(\pmb\lambda^*)$ is not positive
semidefinite and $\pmb{\lambda}^{*}$ is not APPT.  The decreasingly ordered
APPT spectral set is compact.  If its maximum purity equalled the
outer-polytope maximum, an APPT spectrum would therefore attain that value
and, by uniqueness, would have to be $\pmb{\lambda}^{*}$.  This is impossible,
so the upper bound is strict.

It remains to verify sharpness in qubit--qudit systems.  For every $n\ge4$,
Theorem~\ref{linear} supplies a maximizing type~$3$ spectrum.  Its three
smallest eigenvalues are equal and
$\lambda_1=\lambda_{2n-1}+2\sqrt{\lambda_{2n-2}\lambda_{2n}}$.
Thus the exact qubit criterion \eqref{ineq} shows that the spectrum is
APPT and hence, by \cite{Jo}, absolutely separable.  This proves
the remaining sharpness cases.

Finally, put $d=mn$.  If $d\equiv0\pmod4$, then
$4/(3d)<2/d$.  In the remaining three cases modulo $4$, multiplication by the
positive denominators reduces the desired inequalities to
$$
3d^2-2d+1>0,\qquad 3d^2+4d+4>0,\qquad 3d^2+2d+1>0,
$$
respectively.  Hence the displayed bound is strictly smaller than $2/(mn)$.
Multiplying each of its four formulas by $mn$ also shows that the product
converges to $4/3$ as $mn\longrightarrow\infty$.
\end{proof}

\section{A universal entropy bound and the exact qubit--qudit minimum}

In this and the next section, $d$ denotes the total number of
eigenvalues.
Recall that the spectral entropy is
$H(\pmb\lambda)=-\sum_j\lambda_j\log_2\lambda_j$.
The vertex description of $\CS_d$ makes it possible to minimize this function
over the outer polytope.  At vertices of types $1$, $2$, and
$3$, respectively, the entropy values are
\begin{align}
H_1(i)&=\log_2 i, &&d-2\le i\le d,\label{entropy-type1}\\
H_2(i)&=\log_2(d-1+i)-\frac{2i}{d-1+i},
&&1\le i\le d-3,\label{entropy-type2}\\
H_3(i)&=\log_2(d+2i)-\frac{3i}{d+2i}\log_2 3,
&&1\le i\le d-3.\label{entropy-type3}
\end{align}

\begin{proposition}\label{outer-entropy-prop}
Let $d\ge4$ and put
$$
\widehat t_d=\frac d4(3\ln3-2).
$$
The minimum of $H$ over $\CS_d$ is
$$
\min_{\pmb\lambda\in\CS_d}H(\pmb\lambda)=
\begin{cases}
\log_2(d-2),&4\le d\le14,\\[1mm]
H_3(t_d),&d\ge15,
\end{cases}
$$
where, for $d\ge15$, $t_d$ is the unique minimizer of $H_3$ on
$\{\lfloor\widehat t_d\rfloor,\lceil\widehat t_d\rceil\}$.  Both integers
belong to $\{1,\ldots,d-3\}$.  The minimizing spectrum is unique.  For
$4\le d\le14$, it is
$$
\frac1{d-2}
\bigl(\underbrace{1,\ldots,1}_{d-2},0,0\bigr).
$$
For $d\ge15$, it is
\[
\frac{1}{d+2t_d}
\bigl(\underbrace{3,\ldots,3}_{t_d},
\underbrace{1,\ldots,1}_{d-t_d}\bigr).
\]
\end{proposition}

\begin{proof}
Entropy is strictly concave, so every minimizer on the compact
polytope $\CS_d$ must be a vertex: a nontrivial convex decomposition of a
nonvertex minimizer would contain a point with strictly smaller entropy.
Formulas
\eqref{entropy-type1}--\eqref{entropy-type3} follow by direct
substitution in the three vertex families.  The smallest type $1$
value is $\log_2(d-2)$.  Regarding the other two expressions as
functions of a real variable $x\ge0$, differentiation gives
\begin{align}
H_2'(x)
&=\frac{x-(d-1)(2\ln2-1)}{(d-1+x)^2\ln2},
&\min_{x\ge0}H_2(x)
&=\log_2\!\left(\frac{e\ln2}{2}(d-1)\right),
\label{continuous-H2}\\
H_3'(x)
&=\frac{4(x-\widehat t_d)}{(d+2x)^2\ln2},
&\min_{x\ge0}H_3(x)
&=\log_2\!\left(\frac{e\ln3}{2\sqrt3}\,d\right).
\label{continuous-H3}
\end{align}
For $4\le d\le14$,
$$
\frac{d-2}{d-1}\le\frac{12}{13}<\frac{15}{16}<\frac{e\ln2}{2},
\qquad
\frac{d-2}{d}\le\frac67<\frac{e\ln3}{2\sqrt3}.
$$
Thus the two continuous minima in \eqref{continuous-H2} and
\eqref{continuous-H3} are strictly larger than $\log_2(d-2)$.  Hence the
type~$1$ minimum is strictly smallest for $4\le d\le14$.

It remains to locate the transition for $d\ge15$.  Write
$d=3q+s$, where $q\ge5$ and $s\in\{0,1,2\}$, and evaluate the type $3$
expression at the allowed index $i=q$.  For the following monotonicity
argument, temporarily allow $q\ge5$ and $0\le s\le2$ to be real variables.
In natural logarithms set
$$
\Psi(q,s)
=\ln(5q+s)-\frac{3q}{5q+s}\ln3.
$$
Because $e\ln2/2>15/16$, \eqref{continuous-H2} shows that every type $2$
entropy is greater than $\log_2(15(d-1)/16)$.  To compare $H_3(q)$ with
this lower bound and with the type $1$ minimum, set
$$
D_1(q,s)
=\Psi(q,s)-\ln(3q+s-2),
$$
and
$$
D_2(q,s)
=\Psi(q,s)
-\ln\!\left(\frac{15(3q+s-1)}{16}\right).
$$
Direct differentiation gives
\begin{align*}
\partial_qD_1
&=\frac{2s-10}
{(5q+s)(3q+s-2)}
-\frac{3s\ln3}{(5q+s)^2},\\
\partial_sD_1
&=-\frac{2q+2}{(5q+s)(3q+s-2)}
+\frac{3q\ln3}{(5q+s)^2},\\
\partial_qD_2
&=\frac{2s-5}
{(5q+s)(3q+s-1)}
-\frac{3s\ln3}{(5q+s)^2},\\
\partial_sD_2
&=-\frac{2q+1}{(5q+s)(3q+s-1)}
+\frac{3q\ln3}{(5q+s)^2}.
\end{align*}
The $q$-derivatives are immediately negative.  Using $\ln3<10/9$ for
the $s$-derivatives gives
$$
\partial_sD_1<
\frac{4qs-50q-6s}
{3(5q+s)^2(3q+s-2)}<0,
\qquad
\partial_sD_2<
\frac{4qs-25q-3s}
{3(5q+s)^2(3q+s-1)}<0
$$
for $q\ge5$ and $0\le s\le2$.  At
$(q,s)=(5,0)$, the values are
$$
D_1(5,0)=\frac15\ln\!\left(\frac{25^5}{27\cdot13^5}\right)<0,
\qquad
D_2(5,0)=\frac15\ln\!\left(\frac{40^5}{27\cdot21^5}\right)<0.
$$
The final inequalities are the integer comparisons
$25^5<27\cdot13^5$ and $40^5<27\cdot21^5$.  Since both partial
derivatives are negative on $q\ge5$, $0\le s\le2$, we have
$D_j(q,s)\le D_j(5,0)<0$ for $j=1,2$.  Thus the type $3$
vertex with $i=q$ has smaller entropy than every type $1$ and type $2$
vertex whenever $d\ge15$.

Finally, the derivative in \eqref{continuous-H3} shows that $H_3$ decreases
up to $\widehat t_d$ and increases thereafter.  Moreover,
$d/4<\widehat t_d<d/3$ follows from
$1<\ln3<10/9$.  Thus the two candidate integers are allowed for
$d\ge15$.  Hence the minimum of $H_3(i)$ over integers is attained at
$\lfloor\widehat t_d\rfloor$ or $\lceil\widehat t_d\rceil$.  There can be no tie
between adjacent integers $i$ and $i+1$: equality would imply
$$
\frac{d+2i+2}{d+2i}
=3^{\frac{3d}{(d+2i)(d+2i+2)}},
$$
where the exponent lies strictly between $0$ and $1$.  The left-hand side
is rational, whereas the
right-hand side is irrational, since it is a nonintegral rational
power of $3$.
Hence the minimizing integer, and therefore the minimizing vertex, is
unique.
\end{proof}

Since every ordered APPT spectrum belongs to $\CS_{mn}$, Proposition
\ref{outer-entropy-prop} immediately gives the following explicit
bound in every bipartite dimension.  Write
$$
\begin{aligned}
S_{\min}^{\mathrm{APPT}}(m,n)
&=\min_{\pmb\lambda\in(\Lambda^{\mathrm{APPT}}_{m,n})^\downarrow}
H(\pmb\lambda),\\
S_{\min}^{\mathrm{ASEP}}(m,n)
&=\min_{\pmb\lambda\in(\Lambda^{\mathrm{ASEP}}_{m,n})^\downarrow}
H(\pmb\lambda).
\end{aligned}
$$
Then
\begin{equation}\label{universal-entropy-lower-bound}
\begin{aligned}
S_{\min}^{\mathrm{ASEP}}(m,n)
&\ge S_{\min}^{\mathrm{APPT}}(m,n),\\
S_{\min}^{\mathrm{APPT}}(m,n)
&\ge
  \begin{cases}
  \log_2(mn-2),&4\le mn\le14,\\[1mm]
  \displaystyle
  \log_2(mn+2t_{mn})
  -\frac{3t_{mn}}{mn+2t_{mn}}\log_2 3,&mn\ge15.
  \end{cases}
\end{aligned}
\end{equation}
Here $t_{mn}$ is the integer specified in Proposition
\ref{outer-entropy-prop}.  In particular, this entropy bound uses only
the defining inequality of $\CS_{mn}$ and does not involve an inner
polytope.

We next use Song and Chen's description of the extreme qubit--qudit
spectra to treat the remaining dimensions $2\le n\le7$, for which the
minimizer over $\CS_{2n}$ does not satisfy the exact qubit criterion.

\begin{lemma}[Endpoint reduction]
\label{song-chen-entropy-lemma}
Let $d=2n$ with $2\le n\le7$, and put
$$
\theta=3+2\sqrt2.
$$
The minimum of $H$ over the decreasingly ordered qubit--qudit APPT
spectra is attained at one of
\begin{equation*}
\begin{aligned}
U_d&=\frac{1}{d-1}
\bigl(\underbrace{1,\ldots,1}_{d-1},0\bigr),\\
V_d&=\frac{1}{(d-2)\theta+2}
\bigl(\underbrace{\theta,\ldots,\theta}_{d-2},
\underbrace{1,\ldots,1}_{2}\bigr),\\
T_{d,k}&=\frac{1}{d+2k}
\bigl(\underbrace{3,\ldots,3}_{k},
\underbrace{1,\ldots,1}_{d-k}\bigr),
\qquad 1\le k\le d-3.
\end{aligned}
\end{equation*}
Their entropies are, respectively,
\begin{align}
u_d:=H(U_d)&=\log_2(d-1),\label{ud-def}\\
v_d:=H(V_d)&=\log_2((d-2)\theta+2)
-\frac{(d-2)\theta}{(d-2)\theta+2}\log_2\theta,
\label{vd-def}\\
\tau_d(k):=H(T_{d,k})=H_3(k)&=\log_2(d+2k)
-\frac{3k}{d+2k}\log_2 3.\label{taud-def}
\end{align}
\end{lemma}

The proof of Lemma~\ref{song-chen-entropy-lemma} is deferred to
Appendix~\ref{appendix-song-chen-endpoint}.

\begin{proof}[Proof of Theorem \ref{entropy-main-thm}]
\smallskip
\noindent\emph{The case $n\ge8$.}
Put $d=2n$.  Since $\widehat t_{2n}=\frac n2(3\ln3-2)$, Proposition
\ref{outer-entropy-prop} shows that, for $n\ge8$, $t_{2n}=r_n$ and the
minimum over $\CS_{2n}$ is $h_n(r_n)$.  Its unique minimizing spectrum
is
$$
\pmb\lambda^*=\frac{1}{2n+2r_n}
\bigl(\underbrace{3,\ldots,3}_{r_n},
\underbrace{1,\ldots,1}_{2n-r_n}\bigr).
$$
The three smallest entries of $\pmb\lambda^*$ are all
$1/(2n+2r_n)$, while its largest entry is $3/(2n+2r_n)$.  Hence
$$
\lambda^*_1
=\lambda^*_{2n-1}
+2\sqrt{\lambda^*_{2n-2}\lambda^*_{2n}},
$$
so $\pmb\lambda^*$ satisfies the exact qubit APPT criterion
\eqref{ineq}.  It therefore belongs to
$(\Lambda^{\mathrm{APPT}}_{2,n})^\downarrow$.  Combining this
fact with the second inequality in
\eqref{universal-entropy-lower-bound} proves that the APPT minimum is
$h_n(r_n)$.  Johnston's equality of the APPT and absolutely separable
spectral sets in qubit--qudit systems \cite{Jo} gives the same result
for absolute separability.  Uniqueness follows from Proposition
\ref{outer-entropy-prop}: the decreasingly ordered minimizer is unique,
or equivalently the spectrum is unique up to permutation.

\smallskip
\noindent\emph{The cases $2\le n\le7$.}
Lemma
\ref{song-chen-entropy-lemma} reduces the problem to a finite endpoint
comparison.  Appendix~\ref{appendix-entropy-comparisons} gives exact
rational-interval certificates showing that the unique minimizing endpoints
for $d=4,6,8,10,12,14$ are, respectively,
$U_4,V_6,T_{8,3},T_{10,3},T_{12,4},T_{14,5}$.  It also identifies the
closest competitor and a positive rational lower bound for each entropy gap.
In particular,
$r_4=3$, $r_5=3$, $r_6=4$, and $r_7=5$.
The minimizing endpoints give exactly the spectra and values stated in the
theorem.  Appendix~\ref{appendix-song-chen-endpoint} shows that each
one-parameter family is monotone or increases and then decreases.  Thus
the positive endpoint gaps also give uniqueness of the decreasingly ordered
minimizing spectrum in these dimensions.
\end{proof}

\begin{remark}
The $2\otimes3$ spectrum confirms the numerical entropy candidate of
Song and Chen.  Their reported $2\otimes4$ candidate
\cite[Remark following Corollary~11]{SC} corresponds to $T_{8,4}$,
whereas the exact minimizer is $T_{8,3}$; indeed,
$$
H(T_{8,4})-H(T_{8,3})
=\log_2\frac87-\frac3{28}\log_2 3>0,
$$
the last inequality being equivalent to $(8/7)^{28}>27$; this follows
from $(8/7)^7>5/2$ and $(5/2)^4>27$.
\end{remark}

\section{Outer spectral volume and qubit--qudit decay}

Let
$$
\Delta_{d-1}=\{\pmb\lambda\in\mathbb R^d:
\lambda_i\ge0,\ \textstyle\sum_i\lambda_i=1\}
$$
and let $\pmb\lambda^\downarrow$ denote the vector obtained by arranging the
entries of $\pmb\lambda$ in decreasing order.  We use the usual
$(d-1)$-dimensional Euclidean volume on this simplex and define
$$
\widehat{\CS}_d
=\{\pmb\lambda\in\Delta_{d-1}:
\pmb\lambda^\downarrow\in\CS_d\},\qquad
R_d=\frac{\operatorname{vol}_{d-1}(\widehat{\CS}_d)}
{\operatorname{vol}_{d-1}(\Delta_{d-1})}.
$$

\begin{proposition}\label{outer-volume-prop}
For every $d\ge4$,
\begin{equation}\label{outer-volume-beta}
R_d=\frac{d(d-1)}{d-2}
\left[\BetaFn\!\left(\frac d2,d-2\right)
-\BetaFn(d-1,d-2)\right],
\end{equation}
where $\BetaFn(x,y)=\Gamma(x)\Gamma(y)/\Gamma(x+y)$.
\end{proposition}

\begin{proof}
Recall from Lemma~\ref{outer-vertices-lemma} that
$\CO_d=\operatorname{conv}\{\pmb\nu_1,\ldots,\pmb\nu_d\}$.  The affine map
$$
\pmb x\longmapsto\sum_i x_i\pmb\nu_i
$$
gives a one-to-one correspondence from $\Delta_{d-1}$ onto $\CO_d$.  Because it multiplies every
$(d-1)$-dimensional volume by the same nonzero constant, it preserves volume
ratios.  For
$G(\pmb\lambda)=\lambda_{d-2}+\lambda_{d-1}+\lambda_d-\lambda_1$,
put $g_i=G(\pmb\nu_i)$.  The calculation in
Lemma~\ref{outer-vertices-lemma} gives
$$
g_i=-\frac1i\quad(1\le i\le d-3),\qquad
g_{d-2}=0,\qquad g_{d-1}=\frac1{d-1},\qquad g_d=\frac2d.
$$
The points that this map sends into $\CS_d$ are
$\{\pmb x\in\Delta_{d-1}:\sum_i g_i x_i\ge0\}$.
For $d\ge4$, these coefficients are all distinct; in particular,
$1/(d-1)\ne2/d$.  For distinct real coefficients $z_1,\ldots,z_d$,
scaling Lasserre's simplex-section formula and passing to the complementary
halfspace \cite[Theorem 2.2]{Las} gives
\begin{equation*}
\frac{\operatorname{vol}_{d-1}\{\pmb x\in\Delta_{d-1}:
\sum_i z_ix_i\ge0\}}
{\operatorname{vol}_{d-1}(\Delta_{d-1})}
=\sum_{z_i>0}\frac{z_i^{d-1}}
{\prod_{j\ne i}(z_i-z_j)}.
\end{equation*}
Applying this complementary form with $z_i=g_i$, only $i=d-1,d$ contribute.
The zero coefficient $g_{d-2}=0$ causes no ambiguity, since the formula at
this boundary value follows by continuity in the halfspace threshold.  Using
$\prod_{j=1}^{r}(a+j)=\Gamma(a+r+1)/\Gamma(a+1)$ to simplify the two
products gives
$$
\frac{g_{d-1}^{d-1}}
{\prod_{j\ne d-1}(g_{d-1}-g_j)}
=-\frac{d(d-1)}{d-2}\BetaFn(d-1,d-2)
$$
and
$$
\frac{g_d^{d-1}}
{\prod_{j\ne d}(g_d-g_j)}
=\frac{d(d-1)}{d-2}\BetaFn\!\left(\frac d2,d-2\right).
$$
Their sum proves \eqref{outer-volume-beta}.  Finally, the full simplex is the
union of the $d!$ copies of $\CO_d$ obtained by permuting coordinates.  Their interiors are
disjoint, and their overlaps have volume zero.  The same statement holds for
$\widehat{\CS}_d$, so the ratio computed in $\CO_d$ equals the
permutation-invariant relative-volume ratio.
\end{proof}

For a measurable set $E\subseteq\Delta_{d-1}$, write
$$
\operatorname{rvol}(E)
:=\frac{\operatorname{vol}_{d-1}(E)}
{\operatorname{vol}_{d-1}(\Delta_{d-1})}
$$
for its relative spectral volume.  Thus
$$
a_n=\operatorname{rvol}(\Lambda^{\mathrm{APPT}}_{2,n})
=\operatorname{rvol}(\Lambda^{\mathrm{ASEP}}_{2,n}).
$$

We next compare $a_n$ with explicit inner- and outer-polytope volumes; these
bounds make the large-$n$ decay transparent.

These are Euclidean volumes on the simplex of spectra.  They are different
from Hilbert--Schmidt probabilities on the full state space: after passing to
eigenvalues, the latter use a density containing the factor
$\prod_{i<j}(\lambda_i-\lambda_j)^2$.

For an inner bound with the matching exponential decay rate, we use the
polytope of Ahiable,
Kothakonda and Winter.  Its decreasingly ordered part is
\begin{equation}\label{inner-qubit-polytope}
\CP_{2,n}^{\downarrow}
=\{\pmb\lambda\in\CO_{2n}:
\lambda_1\le\lambda_{2n-1}+2\lambda_{2n}\}.
\end{equation}
Here $\CP_{2,n}$ denotes the full permutation-invariant polytope obtained from
\eqref{inner-qubit-polytope}; see
\cite[Definition 6.1 and Theorem 6.2]{AKW}.  The exact qubit criterion, Johnston's
equivalence \cite{Jo}, and the defining inequality of $\CS_{2n}$ give
\begin{equation}\label{inner-outer-qubit}
\CP_{2,n}^{\downarrow}
\subseteq(\Lambda^{\mathrm{APPT}}_{2,n})^\downarrow
=(\Lambda^{\mathrm{ASEP}}_{2,n})^\downarrow
\subseteq\CS_{2n}.
\end{equation}
Indeed, ordering gives $\lambda_{2n-2}\ge\lambda_{2n}$, so the inner
inequality implies the exact qubit criterion \eqref{ineq}; the outer
inequality then follows from $2\sqrt{ab}\le a+b$.
The quantities $p_n$ and $q_n$ in Theorem~\ref{volume-main-thm} are the
relative volumes of the inner and outer polytopes, respectively.  The formula
for $p_n$ below is the $m=2$ specialization of
\cite[Proposition 6.14]{AKW}:
\begin{equation}\label{inner-outer-volumes}
\begin{aligned}
p_n&=\operatorname{rvol}(\CP_{2,n})
=\frac{n!(2n-2)!}{(3n-2)!},\\
q_n&=R_{2n}=\frac{n(2n-1)}{n-1}
\bigl[\BetaFn(n,2n-2)-\BetaFn(2n-1,2n-2)\bigr].
\end{aligned}
\end{equation}
Equivalently, expanding the beta functions gives
\begin{equation}\label{outer-qubit-factorial}
q_n=
\frac{(2n-1)(2n-3)!n!}{(n-1)(3n-3)!}
-\frac{2n(2n-1)((2n-3)!)^2}{(4n-4)!}.
\end{equation}

\begin{proof}[Proof of Theorem \ref{volume-main-thm}]
\smallskip
\noindent\emph{Finite-dimensional bounds.}
All three spectral sets are invariant under coordinate permutations, so the
ordered inclusions in \eqref{inner-outer-qubit} give inclusions in the full
probability simplex.  These inclusions, together with Proposition
\ref{outer-volume-prop}, give
$$
p_n\le a_n\le q_n.
$$
Dividing \eqref{outer-qubit-factorial} by the formula for $p_n$ in
\eqref{inner-outer-volumes} and simplifying yields $q_n=c_np_n$, where
\begin{equation}\label{volume-ratio-cn}
c_n=\frac{(2n-1)(3n-2)}{2(n-1)^2}
-\frac{(2n-1)(2n-3)!(3n-2)!}
{(n-1)(n-1)!(4n-4)!}.
\end{equation}
The quantity subtracted is positive, and the first term equals
$$
3+\frac{5n-4}{2(n-1)^2}.
$$
Now $c_2=3$ and $c_3=55/16<4$.  For $n\ge4$, the inequality
$2n^2-9n+6>0$ shows that the first term itself is smaller than $4$.
Thus $c_n<4$ for every $n\ge2$, proving all finite-dimensional inequalities
in Theorem~\ref{volume-main-thm}.

\smallskip
\noindent\emph{Asymptotics.}
We use Stirling's formula in the form
$$
k!=\sqrt{2\pi k}\left(\frac{k}{e}\right)^k
\left(1+O\left(\frac1k\right)\right)
\qquad(k\longrightarrow\infty).
$$
Applying it to the formula for $p_n$ in \eqref{inner-outer-volumes} gives
the first asymptotic relation in \eqref{volume-asymptotics} below.  A second
application of Stirling's formula shows that the positive second term on the
right-hand side of \eqref{volume-ratio-cn} is asymptotic to
$3\sqrt{3/2}\,(27/64)^{n-1}$ and therefore tends to zero.  Since
the first term in \eqref{volume-ratio-cn} is
$3+(5n-4)/[2(n-1)^2]$, it follows that $c_n\to3$.  Therefore
$q_n=c_np_n$ and
\begin{equation}\label{volume-asymptotics}
p_n\sim\frac32\sqrt{3\pi n}\left(\frac4{27}\right)^n,
\qquad
q_n\sim\frac92\sqrt{3\pi n}\left(\frac4{27}\right)^n.
\end{equation}
The two asymptotic formulas and $p_n\le a_n\le q_n$ first give
$a_n=\Theta(\sqrt n(4/27)^n)$.  Moreover,
$$
-\frac1n\ln q_n\le-\frac1n\ln a_n\le-\frac1n\ln p_n,
\qquad
\lim_{n\to\infty}-\frac1n\ln p_n
=\lim_{n\to\infty}-\frac1n\ln q_n=\ln\frac{27}{4}.
$$
The squeeze theorem therefore gives the asserted exponential decay rate for
$a_n$.
\end{proof}

\section{Conclusion}

Together with the Song--Chen classification, the outer-polytope method gives
sharp qubit--qudit leading-partial-sum bounds for $1\le k\le2n-3$, the exact
minimum entropy, and, when $n\ge3$, the exact maximum purity.  Combined with
the Ahiable--Kothakonda--Winter inner polytope, it also gives two-sided
spectral-volume bounds with exact exponential decay rate $\ln(27/4)$.  The
outer-polytope relaxation is tight for qubit purity when $n\ge4$, for entropy
when $n\ge8$, and for all leading partial sums in the stated range; the
Song--Chen classification supplies the remaining qubit purity and
entropy cases.  By contrast, the relaxation is not tight for qutrit purity,
which disproves the D\~ung--Kh\^oi conjecture.  The volume result further shows
that qubit--qudit APPT spectra occupy an exponentially shrinking fraction of
the spectral simplex.  Determining the exact maximum APPT purity remains open
in general for $m\ge3$.

\section*{Author Declarations}

\noindent\textbf{Conflict of Interest}.  The author has no conflicts to
disclose.

\medskip
\noindent\textbf{Author Contributions}.  Anh T. Tran: Conceptualization,
Formal analysis, Investigation, Methodology, Validation, Writing -- original
draft, Writing -- review and editing.

\medskip
\noindent\textbf{AI Disclosure}.  The author used GPT-5.6 Sol through ChatGPT (OpenAI) to improve the language and to perform secondary checks of algebraic simplifications and rational-interval computations. The purpose was to reduce the risk of transcription and arithmetic errors; the tool was not treated as a source of mathematical authority. Every proof, citation, calculation, and conclusion was independently reviewed and verified by the author, who takes full responsibility for the contents of the manuscript. 

\medskip
\noindent\textbf{Funding}.  This research received no external funding.

\section*{Data Availability}

Data sharing is not applicable to this article as no new data were created or
analyzed in this study.

\appendix

\section{Proof of the endpoint-reduction lemma}
\label{appendix-song-chen-endpoint}

Since the APPT state set is compact and convex and von Neumann entropy is
strictly concave, every entropy minimizer is an extreme point.
Song and Chen describe the extreme points and give their formulas in
\cite[Theorem 9 and Eqs.~(B11)--(B13)]{SC} for $n=2$, and in
\cite[Theorem 10 and Eqs.~(B14)--(B16)]{SC} for $n\ge3$.  Since both the
APPT property and entropy are invariant under unitary conjugation, entropy
depends only on the spectrum.  Using Song and Chen's parameter $r$, the
three families below are their three one-parameter extreme-point families,
rewritten with a common spectral normalization.  The
substitutions $x=r^{-1/2}$ in Eqs.~(B11), (B12), (B14), and (B15), and
$x=\sqrt r$ in Eqs.~(B13) and (B16), give the following normalized spectra;
these three families exhaust the classification.  In the last two families,
$k$ records the indicated eigenvalue multiplicity; for the last family, write
$\ell=d-k-2$.
\begin{equation*}
\begin{aligned}
\mathcal A(x)&=\frac{1}{d-1-2x+x^2}
\bigl(\underbrace{1,\ldots,1}_{d-2},1-2x,x^2\bigr),
&&0\le x\le\sqrt2-1,\\
\mathcal B_k(x)&=
\frac{1}{d-1+2kx+x^2}
\bigl(\underbrace{1+2x,\ldots,1+2x}_{k},
\underbrace{1,\ldots,1}_{d-k-1},x^2\bigr),
&&0\le x\le1,\\
\mathcal C_k(x)&=
\frac{1}{k(1+2x)+\ell x^2+2}
\bigl(\underbrace{1+2x,\ldots,1+2x}_{k},
\underbrace{x^2,\ldots,x^2}_{\ell},
\underbrace{1,\ldots,1}_{2}\bigr),
&&1\le x\le1+\sqrt2,
\end{aligned}
\end{equation*}
where $1\le k\le d-3$.  More explicitly,
$$
\mathcal A(0)=U_d,\quad \mathcal A(\sqrt2-1)=V_d,
\qquad
\mathcal B_k(0)=U_d,\quad \mathcal B_k(1)=T_{d,k},
$$
and
$$
\mathcal C_k(1)=T_{d,k},\qquad
\mathcal C_k(1+\sqrt2)=V_d.
$$
All differentiations below are taken on the interiors of the stated parameter
intervals.  Endpoint values, including those at which a displayed weight
vanishes, and the corresponding conclusions on the closed intervals follow
by continuity.
For each family, we show that the entropy is monotone or increases and then
decreases.  Consequently, it has no interior minimum, and its minimum is
attained at an endpoint.

It is enough to use natural logarithms.  We repeatedly use the
elementary estimates
\begin{equation}\label{elementary-log-estimates}
2u<\ln\frac{1+u}{1-u}
<2\left(u+\frac{u^3}{3(1-u^2)}\right),\qquad 0<u<1,
\end{equation}
and
\begin{equation}\label{elementary-log-estimate-second}
\ln(1+y)<\frac{y(y+2)}{2(y+1)},\qquad y>0.
\end{equation}
Indeed,
$$
\ln\frac{1+u}{1-u}
=2\sum_{j=0}^{\infty}\frac{u^{2j+1}}{2j+1},
\qquad
0<\sum_{j=1}^{\infty}\frac{u^{2j+1}}{2j+1}
<\frac13\sum_{j=1}^{\infty}u^{2j+1}
=\frac{u^3}{3(1-u^2)},
$$
which proves \eqref{elementary-log-estimates}.  For
\eqref{elementary-log-estimate-second}, the difference
$f(y)=y(y+2)/[2(y+1)]-\ln(1+y)$ satisfies $f(0)=0$ and
$f'(y)=y^2/[2(y+1)^2]>0$ for $y>0$.

We shall also use the following derivative identity.  If positive weights
$b_j(x)$ are normalized by $Z(x)=\sum_j b_j(x)$, set
$$
p_j(x)=\frac{b_j(x)}{Z(x)},\qquad
\widetilde H(x):=-\sum_jp_j(x)\ln p_j(x)
=(\ln2)H\bigl((p_j(x))_j\bigr).
$$
Then
\begin{equation}\label{normalized-entropy-derivative}
\widetilde H'
=\frac{Z'\sum_j b_j\ln b_j-Z\sum_j b_j'\ln b_j}{Z^2}.
\end{equation}
It follows by differentiating
$\widetilde H=\ln Z-Z^{-1}\sum_j b_j\ln b_j$.  Formula
\eqref{normalized-entropy-derivative} gives the derivative expressions
analyzed below.  For brevity, write
$$
\eta_0(x)=\ln\frac1x,\qquad
\eta_-(x)=\ln\frac1{1-2x},\qquad
\eta_+(x)=\ln(1+2x),
$$
whenever the corresponding expression is defined.

\smallskip
\noindent\emph{$\mathcal A$-family.}
The entropy derivative has the sign of
$$
(d-2+x-x^2)\ln(1-2x)-2x(d-1-x)\ln x.
$$
Set $q_A=d-2$ and
$$
\Phi_A(x)=
\frac{2x(q_A+1-x)\eta_0(x)}{(q_A+x-x^2)\eta_-(x)}.
$$
Thus the derivative has the sign of $\Phi_A-1$.  If
$Y_A=2x(q_A+1-x)/(q_A+x-x^2)$, then, on
$0<x<\sqrt2-1$, logarithmic differentiation gives
$$
(\ln\Phi_A)'=\frac{Y_A'}{Y_A}-\frac1{x\eta_0(x)}
-\frac2{(1-2x)\eta_-(x)},
$$
and
$$
\frac1x-\frac{Y_A'}{Y_A}
=\frac{(1-x)(2q_A+1-x)}
{(q_A+1-x)(q_A+x-x^2)}>0.
$$
Hence $Y_A'/Y_A<1/x$.  Moreover, the power series with positive terms
$$
\eta_-(x)=-\ln(1-2x)=\sum_{j=1}^{\infty}\frac{(2x)^j}{j}
<\frac{2x}{1-2x}
$$
implies $2/[(1-2x)\eta_-(x)]>1/x$.  Therefore
$$
(\ln\Phi_A)'<\frac1x-\frac1{x\eta_0(x)}-\frac1x
=-\frac1{x\eta_0(x)}<0.
$$
Thus $\Phi_A$ is strictly decreasing.  The entropy derivative therefore
changes sign at most once, and only from positive to negative, so the entropy
has no interior minimum in the family $\mathcal A$.

\smallskip
\noindent\emph{$\mathcal B_k$-family.}
Put $q_B=d-1-k\ge2$.  Its entropy derivative has
the sign of
$$
D_B(x)=k(x^2+x-q_B)\ln(1+2x)
-2x[k(x+1)+q_B]\ln x.
$$
Suppose first that $q_B\ge3$.  Since $q_B-x-x^2>0$ for $0<x<1$,
$D_B(x)$ has the sign of $\Phi_B(x)-1$, where
$$
\Phi_B(x)=
\frac{2x[k(x+1)+q_B]\eta_0(x)}
{k(q_B-x-x^2)\eta_+(x)}.
$$
Logarithmic differentiation gives
\begin{align*}
(\ln\Phi_B)'&=\frac1x+\frac{k}{k(x+1)+q_B}
+\frac{1+2x}{q_B-x-x^2}-\frac1{x\eta_0(x)}
-\frac2{(1+2x)\eta_+(x)}.
\end{align*}
The upper bound
$\eta_+(x)<2x(1+x)/(1+2x)$ from
\eqref{elementary-log-estimate-second} implies
$$
\frac1x-\frac2{(1+2x)\eta_+(x)}<\frac1{1+x}.
$$
Moreover,
$$
\frac{k}{k(x+1)+q_B}<\frac1{1+x},\qquad
\frac{1+2x}{q_B-x-x^2}\le
\frac{1+2x}{3-x-x^2},
$$
so combining the three estimates gives
$$
(\ln\Phi_B)'<
\frac2{1+x}+\frac{1+2x}{3-x-x^2}
-\frac1{x\eta_0(x)}.
$$
Applying the upper bound in \eqref{elementary-log-estimates} with
$u=(1-x)/(1+x)$ gives
$$
\eta_0(x)<\frac{(1-x)(x^2+10x+1)}{6x(1+x)}.
$$
Consequently,
$$
(\ln\Phi_B)'<
\frac{5x^4+2x^3-54x^2+34x-11}
{(x-1)(x+1)(x^2+x-3)(x^2+10x+1)}<0.
$$
Indeed, the denominator is positive on $0<x<1$.  The numerator is smaller
than $-47x^2+34x-11$ because $5x^2+2x-7<0$ there; the latter quadratic
is negative because its leading coefficient is negative and its
discriminant is $-912$.  Hence $\Phi_B$ is strictly decreasing, so $D_B$,
and therefore the entropy derivative, can change sign only from positive to
negative.  There is no interior minimum.

It remains to consider $q_B=2$.  The bounds
$$
\eta_0(x)>\frac{2(1-x)}{1+x},\qquad
\eta_+(x)<\frac{2x(1+x)}{1+2x},
$$
follow from \eqref{elementary-log-estimates} and
\eqref{elementary-log-estimate-second}, respectively.  Substitution
into $D_B$ yields the explicit lower bound
$$
D_B(x)>
\frac{2x(1-x)}{(1+x)(1+2x)}
\bigl[kx(1-x^2)+8x+4\bigr]>0.
$$
Thus in this case the entropy is strictly increasing, and again its
minimum occurs at an endpoint.

\smallskip
\noindent\emph{$\mathcal C_k$-family.}
Recall that $\ell=d-k-2\ge1$.  Its entropy derivative
has the sign of
$$
k[\ell(x^2+x)-2]\ln(1+2x)
-2\ell x[k(x+1)+2]\ln x.
$$
For $\ell\ge2$, this is the sign of $\Phi_C-1$, with
$$
\Phi_C(x)=
\frac{k[\ell(x^2+x)-2]\ln(1+2x)}
{2\ell x[k(x+1)+2]\ln x}.
$$
Its exact logarithmic derivative is
\begin{align*}
(\ln\Phi_C)'={}&
\frac{\ell(2x+1)}{\ell(x^2+x)-2}-\frac1x
-\frac{k}{k(x+1)+2}\\
&+\frac{2}{(1+2x)\ln(1+2x)}-\frac1{x\ln x}.
\end{align*}
For $\ell\ge2$ and $k\ge1$,
$$
\frac{\ell(2x+1)}{\ell(x^2+x)-2}-\frac1x
\le\frac{x^2+1}{x(x^2+x-1)},
\qquad
-\frac{k}{k(x+1)+2}\le-\frac1{x+3}.
$$
Applying \eqref{elementary-log-estimates} with
$u=x/(x+1)$ and $u=(x-1)/(x+1)$ gives, respectively,
$$
\ln(1+2x)>\frac{2x}{x+1},
\qquad
\ln x<\frac{(x-1)(x^2+10x+1)}{6x(x+1)}.
$$
Combining these four estimates yields
\begin{align*}
(\ln\Phi_C)'&<
\frac{x^2+1}{x(x^2+x-1)}-\frac1{x+3}
+\frac{x+1}{x(2x+1)}
-\frac{6(x+1)}{(x-1)(x^2+10x+1)}\\
&=
\frac{N_C(x)}
{(x-1)(x+3)(2x+1)(x^2+x-1)(x^2+10x+1)},
\end{align*}
where
$$
N_C(x)=x^6+6x^5+18x^4-69x^3-78x^2-33x+11.
$$
The denominator is positive for $1<x<1+\sqrt2$.  With $y=x-1$,
$$
N_C(x)=y^6+12y^5+63y^4+83y^3-102y^2-288y-144
<28y^2-74y-144<0,
$$
where $0<y<\sqrt2$; the first strict inequality is equivalent to
$y(2-y^2)(y^3+12y^2+65y+107)>0$.  The final quadratic is convex and is
negative at both endpoints of $[0,\sqrt2]$.  It therefore lies below the
line segment joining its endpoint values and is negative throughout this
interval.  Hence $\Phi_C$ is strictly decreasing, so the entropy derivative
can change sign only from positive to negative.  There is no interior
minimum.

\begin{samepage}
\noindent\mbox{It remains to treat $\ell=1$.  Set}
$$
J(x)=2x(x+1)\ln x-(x-1)(x+2)\ln(1+2x).
$$
\end{samepage}
Here $J(1)=0$, $J'(1)=4-\ln27$, and $J''(1)=4-2\ln3$.
The bound \eqref{elementary-log-estimate-second} with $y=2$ gives
$\ln3<4/3$, so both displayed derivatives are positive.  Moreover,
$$
J'''(x)=
\frac{2(8x^4+8x^3+16x^2-4x-1)}{x^2(2x+1)^3}>0.
$$
Indeed, the derivative of the polynomial in the numerator is
$32x^3+24x^2+32x-4>0$ for $x\ge1$, and the polynomial has value $27$ at
$x=1$.  Hence $J'''>0$.  Starting with the displayed values at $x=1$ then
gives $J''>0$, $J'>0$, and $J>0$ for $x>1$.  In this case
the sign expression for the entropy derivative is exactly
$$
-kJ(x)-4x\ln x<0.
$$
Every family therefore takes its minimum at an endpoint.
Substitution gives
\eqref{ud-def}--\eqref{taud-def}.

\section{Exact rational verification of the finite entropy comparisons}
\label{appendix-entropy-comparisons}

This appendix verifies the six finite endpoint orderings used in the proof
of Theorem~\ref{entropy-main-thm}.  The calculation uses rational interval
arithmetic only; no floating-point evaluation enters the proof.

For an integer $N\ge0$ and $1\le y\le2$, set
$$
w=\frac{y-1}{y+1},\qquad
L_N(y)=2\sum_{j=0}^{N}\frac{w^{2j+1}}{2j+1},\qquad
E_N(y)=\frac{2w^{2N+3}}{(2N+3)(1-w^2)}.
$$
Since $y=(1+w)/(1-w)$ and $0\le w\le1/3$, the convergent expansion
$$
\ln y=\ln\frac{1+w}{1-w}
=2\sum_{j=0}^{\infty}\frac{w^{2j+1}}{2j+1}
$$
has a nonnegative remainder satisfying
$$
0\le \ln y-L_N(y)
=2\sum_{j=N+1}^{\infty}\frac{w^{2j+1}}{2j+1}
\le\frac{2}{2N+3}\sum_{j=N+1}^{\infty}w^{2j+1}
=E_N(y).
$$
Consequently,
\begin{equation}\label{appendix-log-bound}
L_N(y)\le\ln y\le L_N(y)+E_N(y).
\end{equation}
Write
$$
\mathcal L_N(y)=[L_N(y),L_N(y)+E_N(y)].
$$
For any positive rational $x$, write $x=2^b y$ with
$b\in\mathbb Z$ and
$1\le y<2$.
Applying \eqref{appendix-log-bound} to $y$ and to $2$ gives the rational
interval
$$
\mathcal J_N(x)=b\mathcal L_N(2)+\mathcal L_N(y),
$$
which contains
$$
\ln x=b\ln2+\ln y.
$$
If $I=[x_-,x_+]$ is a positive rational interval, monotonicity gives an
interval whose lower endpoint is the lower endpoint of
$\mathcal J_N(x_-)$ and whose upper endpoint is the upper endpoint of
$\mathcal J_N(x_+)$.  We denote this interval by $\mathcal J_N(I)$.
All additions, subtractions, products, and quotients below are the standard
interval operations; hence every endpoint remains rational.

We now take $N=8$.  Thus, for $1\le y\le2$ and
$w=(y-1)/(y+1)$,
$$
0\le \ln y-2\sum_{j=0}^{8}\frac{w^{2j+1}}{2j+1}
\le \frac{2w^{19}}{19(1-w^2)}.
$$
We also use
$$
s_{\mathrm{lo}}=\frac{7071}{5000}=1.4142,
\qquad
s_{\mathrm{hi}}=\frac{14143}{10000}=1.4143.
$$
These bounds enclose $\sqrt2$, since
$$
2-s_{\mathrm{lo}}^2=\frac{959}{25000000}>0,
\qquad
s_{\mathrm{hi}}^2-2=\frac{24449}{100000000}>0.
$$
Hence
$$
\mathcal I_\theta=
\left[\frac{14571}{2500},\frac{29143}{5000}\right]
$$
contains $\theta=3+2\sqrt2$.

Substituting $\mathcal I_\theta$ into \eqref{vd-def} and using
$\mathcal J_8$ for every logarithm in
\eqref{ud-def}--\eqref{taud-def} gives rational intervals for all
candidate entropies.  Division by the positive interval
$\mathcal J_8(2)$ converts the bounds from natural logarithms to bits.
The same calculation gives
$$
\widehat t_6\in(1,2),\quad
\widehat t_8\in(2,3),\quad
\widehat t_{10}\in(3,4),\quad
\widehat t_{12}\in(3,4),\quad
\widehat t_{14}\in(4,5).
$$
For $d\in\{4,6,8,10,12,14\}$, let
$$
\mathcal X_d=\{U_d,V_d\}\cup\{T_{d,k}:1\le k\le d-3\}.
$$
For $d=4$ there is only one $T$-candidate.  For
$d\in\{6,8,10,12,14\}$, the displayed locations of $\widehat t_d$, together
with the decrease of $\tau_d(k)$ before $\widehat t_d$ and its increase
afterward, show that among the $T$-candidates only the nearest competitors on
either side need be checked.  Applying the rational interval substitution to these candidates
and to $U_d,V_d$ identifies the minimizing endpoint $M_d$ and the closest
competing endpoint $X_d$.  The interval comparisons also certify that the
upper endpoint of the interval enclosing $H(X_d)-H(M_d)$ is smaller than the
lower endpoint of the corresponding interval for every remaining endpoint.
Finally, choosing a positive rational number below the certified lower endpoint for
$H(X_d)-H(M_d)$ gives the displayed gap certificate $\gamma_d$, so that
$$
0<\gamma_d<H(X_d)-H(M_d)
=\min_{X\in\mathcal X_d\setminus\{M_d\}}
\bigl(H(X)-H(M_d)\bigr).
$$
The resulting compact certificate with exact rational entries is
$$
\begin{array}{c@{\qquad}c@{\qquad}c@{\qquad}c}
d&M_d&\text{closest }X_d&\gamma_d\\ \hline
4&U_4&V_4&15798272292/10^{12}\\
6&V_6&U_6&2140026517/10^{12}\\
8&T_{8,3}&T_{8,2}&4030792107/10^{12}\\
10&T_{10,3}&T_{10,4}&4824739822/10^{12}\\
12&T_{12,4}&T_{12,3}&6493155455/10^{12}\\
14&T_{14,5}&T_{14,4}&545679265/10^{12}.
\end{array}
$$
Every entry in the last column is positive, so the candidates $M_d$ are the
unique entropy-minimizing endpoints.  This proves all endpoint comparisons
used in the proof of Theorem~\ref{entropy-main-thm}.

\end{document}